\documentclass[reqno]{amsart}

\usepackage{amsmath,amsfonts,amssymb,amsthm,mathrsfs}
\usepackage[svgnames]{xcolor}
\usepackage[colorlinks,linkcolor=blue,citecolor=blue,urlcolor=Navy,breaklinks]{hyperref}
\numberwithin{equation}{section}
\usepackage{booktabs,tabularx,array,enumitem,xurl,multirow}

\input{source/macros}
\newcommand{\papertitle}{Subgroup Packing for Batched PASTA Transciphering}
\newcommand{\papershorttitle}{Subgroup Packing for PASTA}

\title[\papershorttitle]{\papertitle}
\date{}
\hypersetup{
  pdftitle={\papertitle},
  pdfauthor={Mugurel Barcau; Vicenţiu Paşol; George C. Ţurcaş}
}

\author[M. Barcau \and V. Pa\c sol]{Mugurel Barcau \and Vicen\c tiu Pa\c sol}
\address{Institute of Mathematics of the Romanian Academy \and CertSIGN, Bucharest}
\email{mugurel.barcau@imar.ro; vicentiu.pasol@imar.ro}
\author[G. C. \c Turca\c s]{George C. \c Turca\c s}
\address{Babe\c s-Bolyai University, Cluj-Napoca \and certSIGN, Bucharest}
\email{george.turcas@ubbcluj.ro}
\thanks{This work was supported by certSign Research and Innovation.}

\begin{document}
\begin{abstract}
With transciphering, a server converts symmetrically encrypted records into
homomorphic ciphertexts without learning the records or the symmetric key.
For the PASTA cipher, this conversion involves dense linear maps whose implementation
depends on how record words are arranged in the ciphertext. We ask whether
rearranging a fixed batch can reduce its conversion cost. Instead of storing
each record's words in a contiguous block, our layout interleaves records
so that cyclic word shifts preserve each record's
positions, which form a coset of a cyclic subgroup. For direct evaluation as
a sum of masked translations, we characterize the required displacements and
relate their counts, 255 for the contiguous layout and 128 for the subgroup
layout, to a prior
transversal-difference invariant. This result does not lower-bound arbitrary
composed circuits. We implement three equally batched schedules for complete
PASTA-3 conversion and subsequent public subset-sum queries in HElib. Across
twelve paired corpora under six homomorphic keys, all 24 direct and subgroup
conversions and all 48 subsequent queries return the expected values and
required zeroes. The median paired ratio of direct to subgroup server cost is
1.60, including fresh public generation, conversion and two queries. The
reduction comes with less remaining noise capacity. These results establish
a packing-dependent cost--noise tradeoff.
\end{abstract}

\maketitle
\section{Introduction}
\label{sec:introduction}

Consider a producer that collects confidential records, each containing 128
measurements. The records belong to one owner, who wants to store them on a
server and choose computations after collection. An authorized analyst might
later select some measurements and request their sum for each record. The
server should learn neither the measurements nor the answers; the owner or
its analyst holds the decryption key. Choosing the subset after collection
matters: if the subset were fixed in advance, the producer could compute the
sum itself before encrypting.

Homomorphic encryption (HE) lets the server compute on encrypted records.
The producer can encrypt every batch directly under HE, or use
\emph{transciphering}, which assigns more work to the server. In the latter
approach, the producer uploads symmetrically encrypted records with public
nonce and counter identifiers. Once per period of symmetric-key reuse, it
also supplies an HE encryption of that key. The server evaluates symmetric
decryption under HE, obtaining encrypted records without learning their
contents or the symmetric key. The HE secret key stays with the owner or
analyst. PASTA, a stream cipher over a prime field, was designed for this
hybrid setting~\cite{PASTA}. Compact
symmetric ciphertexts can reduce communication after key provisioning is
amortized, but conversion costs server time; a producer-compute advantage
must also be measured rather than assumed.

Our question is whether the same number of PASTA records can be converted more
cheaply by rearranging their words in a ciphertext. We study original PASTA-3
over $\F_{65537}$ using HElib and the Brakerski--Gentry--Vaikuntanathan (BGV)
HE scheme~\cite{BrakerskiGentryVaikuntanathan2012}. The cipher's two 128-word
state branches undergo four public affine layers, with branch mixing and nonlinear
steps between them. A dense affine layer combines many encrypted words.
Packed HE places values in \emph{slots} on which arithmetic acts in parallel,
using single instruction, multiple data (SIMD) operations
\cite{SmartVercauteren2014SIMD,HaleviShoupAlgorithmsInHElib}. Combining these
values also requires rotations that move them between slots. A \emph{schedule}
is the chosen sequence of such homomorphic operations. Baby-step/giant-step
(BSGS) schedules share rotations across matrix terms
\cite{HaleviShoup2018LinearTransforms}; their costs depend on the placement
of the words.

\subsection{Four records, four words}
\label{sec:four-by-four}

Take a cyclic row of sixteen slots, showing one state branch for four records
$a,b,c,d$. Subscripts identify words. Two arrangements of the same values are
\[
\begin{array}{@{}l@{\quad}llll@{}}
\text{contiguous:}&a_0\ a_1\ a_2\ a_3&b_0\ b_1\ b_2\ b_3&c_0\ c_1\ c_2\ c_3&d_0\ d_1\ d_2\ d_3\\[1mm]
\text{interleaved:}&a_0\ b_0\ c_0\ d_0&a_1\ b_1\ c_1\ d_1&a_2\ b_2\ c_2\ d_2&a_3\ b_3\ c_3\ d_3.
\end{array}
\]
Suppose each destination needs the preceding word of its own record,
cyclically: word zero should receive word three. A right rotation by four
slots in the interleaved arrangement gives
\[
 a_3\ b_3\ c_3\ d_3\mid a_0\ b_0\ c_0\ d_0\mid
 a_1\ b_1\ c_1\ d_1\mid a_2\ b_2\ c_2\ d_2.
\]
Every value stays in its record's positions, including at the row boundary.
In the contiguous arrangement, a right rotation by one gives
\[
 d_3\ a_0\ a_1\ a_2\mid a_3\ b_0\ b_1\ b_2\mid
 b_3\ c_0\ c_1\ c_2\mid c_3\ d_0\ d_1\ d_2.
\]
The first slot of each block now contains another record's word. A left
rotation by three supplies the correct words at those four positions. If
$R$ and $L$ are the two rotated vectors, and the public mask $M$ is one at
block starts and zero elsewhere, then $R+M\odot(L-R)$ is the desired local
rotation. Here $\odot$ means multiplication slot by slot. This repair uses
two row rotations and one public mask product; BSGS can share repaired
rotations across many matrix terms.

Interleaving puts record $b$, word $w$ at $b+4w$. Its positions form the coset
$b+\{0,4,8,12\}$ of the order-four subgroup of the cyclic row. Adding four
changes the word while preserving the record. We call a packing
\emph{subgroup aligned} when word motion follows such a subgroup: each
record occupies one coset, which need not itself be a subgroup. At full size,
we compare $128b+w$ with $b+128w$ in a row of 16,384 slots, for 128 records
of 128 words. Both layouts batch the same number of states, one PASTA counter
per record; we call such layouts \emph{equally batched}. Every counter
starts from the same 256-word PASTA key, repeated in its positions; each has
its own message, public matrices and constants. Coefficient masks preserve
these counter-specific differences while sharing the rotation schedule.

\subsection{What the algebra explains, and what we compare}

We compare direct contiguous BSGS, contiguous BSGS using the local rotations
above, and subgroup-aligned cyclic BSGS. All implement the same complete
PASTA-3 computation, with a required output of message words and zeroes in
unused branch slots. Subsequent queries accept either native placement, so
this workload needs no homomorphic conversion between layouts.

For a direct sum of masked translations---one masked translation per
displacement, evaluated without factoring---its \emph{support} is the set of
distinct displacements, or labels, with nonzero coefficient masks. The
transversal-difference invariant of Barcau, Pa\c{s}ol and
\c{T}urca\c{s}~\cite{BarcauPasolTurcas2026TDN} limits how much choosing
different coset representatives can reduce that support. We connect the
invariant to counter-dependent matrix evaluation: dense contiguous word
maps require 255 direct labels, while exchanging the word and counter axes
gives 128 subgroup labels. This exchange chooses no better coset representatives
for the original logical word coordinate, and it leaves the invariant unchanged. The
result concerns direct sums; it does not lower-bound factored circuits or runtime.

Interleaving, BSGS and layout-aware HE compilation are established
\cite{HaleviShoup2018LinearTransforms,CHET,Porcupine}. Our contribution is the
exact direct-matrix interpretation, explicit equally batched PASTA programs
with transported counter-specific coefficients, and a comparison through
complete conversion and subsequent queries. Together these connect the algebraic
criterion to mask preparation and noise costs that support counts do not predict.

Across twelve paired corpora under six HE keys, the selected direct and
subgroup programs recover every message and required zero slot and answer
two later public subset-sum queries per conversion. The median paired
direct/subgroup server-cost ratio is 1.60, including fresh public generation,
conversion and both queries. The subgroup layout retains less noise margin, and the
evidence does not establish faster subsequent queries. The separate shallow
\emph{direct-ingestion} baseline encrypts records without PASTA; it fails
during its first query and leaves the complete ingestion comparison open.
That failure is distinct from the successful contiguous-direct
\emph{transciphering} comparator used in the paired result.

\section{PASTA-3 and Packed Evaluation}
\label{sec:background}

\subsection{The complete computation}

PASTA is a stream cipher over a prime field designed for hybrid homomorphic
encryption~\cite{PASTA}. We use original PASTA-3 over $\F_{65537}$ and the
public-material generator and known-answer vectors from the authors'
artifact~\cite{PASTAArtifact,PASTAFramework}. The state starts from one
256-word key, split into $x_L,x_R\in\F_{65537}^{128}$. Every counter starts
from this same key. A public nonce and counter initialize the generator; for
each of four layers it draws the left matrix, the right matrix, and then the
constants, in the original order. The nonce/counter pair must not repeat under
the same symmetric key. The identifiers are public, as are the resulting
matrices and constants; the state on which they act remains encrypted.

Write the affine operation at layer $j$ as
\[
 a_L=M_L^{(j)}x_L+c_L^{(j)},\qquad
 a_R=M_R^{(j)}x_R+c_R^{(j)}.
\]
The branch Mix then produces
\[
 (x_L,x_R)=(2a_L+a_R,\;a_L+2a_R).
\]
After each of the first two affine/Mix layers, each branch undergoes
\[
 F(x)_0=x_0,\qquad F(x)_i=x_i+x_{i-1}^{2}\quad(1\le i<128).
\]
The right-hand side uses the pre-transformation state, so updates do not cascade
from one word to the next. After the third affine/Mix layer, every word is
cubed. A fourth affine/Mix layer follows. Its left branch is the 128-word
keystream $z_b$ for counter $b$. For message $v_b$, the public symmetric
ciphertext is $c_b=v_b+z_b$, and transciphering computes $c_b-z_b$ under HE.
This specifies four affine/Mix layers, two predecessor-square transformations,
and one componentwise cube in the original order.

The evaluator receives public material, the encrypted repeated key, and
HE evaluation keys. The clear symmetric key, HE secret key, and expected
outputs are used only for fixture generation or checking. The evaluator
carries the ciphertext through the entire graph without intermediate
decryption or oracle replacement. Ciphertext products require
relinearization material as well as the affine schedules' automorphism keys.

For the full batch there are 16,384 message words. Let $\iota(b,w)$ be the row
coordinate for the chosen layout. Physical slot $2\iota(b,w)+\epsilon$ holds
branch $\epsilon$, with zero denoting left and one denoting right. The returned ciphertext must satisfy
\[
 \mathrm{out}_{2\iota(b,w)}=v_{b,w},\qquad
 \mathrm{out}_{2\iota(b,w)+1}=0
 \quad(0\le b,w<128).
\]
There are 32,768 slots in all. Selecting the left branch with a public mask,
negating it, and adding the public symmetric ciphertext encoded only in that
branch enforces the formula. The zeroes are part of the homomorphic output;
the mask and final subtraction count as evaluation work.

\subsection{Data ownership and the server's view}

The producer and analyst act for one owner, who keeps the HE secret key.
The server receives public/evaluation keys, the encrypted symmetric key,
PASTA ciphertexts, metadata, and public queries. Batch size and layout are
also public. We assume the server follows the circuit while attempting to
learn the data; HE alone does not authenticate malicious computation.
Clearing non-output slots specifies the returned values, not query access
control: the owner can decrypt any ciphertext it obtains. We add neither
multi-owner keys nor a verification protocol.

\subsection{Slots, rotations, and matrix schedules}

SIMD packing represents many field values in one plaintext; homomorphic
arithmetic acts on the corresponding slots in parallel
~\cite{SmartVercauteren2014SIMD,HaleviShoupAlgorithmsInHElib}. A ring
automorphism induces a structured slot permutation. We model the relevant
coordinates by a finite abelian group $G$, with positive translation
\[
 (\tau_g v)_z=v_{z-g}.
\]
A value at source $z-g$ therefore arrives at destination $z$. On a cyclic row
this is a right rotation. Branch swap is a separate permutation.

A direct diagonal implementation of a linear map is
\[
 L=\sum_{g\in S}D_g\tau_g,
\]
where $D_g$ multiplies by a public vector of coefficients. The set $S$ records
which translations have nonzero masks. BSGS factors these labels into baby and
giant steps so that rotations and partial sums can be reused
~\cite{HaleviShoup2018LinearTransforms}. Consequently the number of diagonals
and the number of scheduled rotations are different quantities. An individual
rotation request may in turn use one or more backend key switches, depending
on the available evaluation keys. We keep support sizes, program operations,
key material, remaining noise capacity, and runtime distinct.

\subsection{The fixed quotient geometry}
\label{sec:quotient-geometry}

For contiguous placement, the following notation separates the record/counter
coordinate $H$ from the word coordinate $Q$.
Let $H\le G$, let $Q=G/H$, and let $\pi:G\to Q$ be the projection. A section
$s:Q\to G$ chooses representatives satisfying $\pi\circ s=\mathrm{id}_Q$;
its image $T=s(Q)$ is a transversal, with difference set
$T-T=\{t-t':t,t'\in T\}$. Each slot has a unique coordinate
$h+s(y)$ with $h\in H$ and $y\in Q$. For a layout whose word coordinate follows $s(Q)$,
a within-state movement from word $x$ to word $y$ has label $s(y)-s(x)$.
For contiguous PASTA packing, $G=C_{16384}$, $H=\langle128\rangle$,
$h=128b$ identifies counter $b$, and $s(w)=w$ selects word $w$. Thus $h$
indexes packed counters and the section indexes words within each counter.
The prior invariant is
\[
 \delta(G,H)=\min_{T\text{ a transversal of }G/H}|T-T|.
\]
Translation of the whole transversal leaves its difference set unchanged, so
we may normalize $s(0)=0$. The invariant belongs to the fixed pair $(G,H)$;
$|T-T|$ belongs to a particular choice of representatives. The cyclic formula
and its general mathematical development are prior work
~\cite[Theorem~3.5 and Corollary~3.6]{BarcauPasolTurcas2026TDN}. The next
section establishes the precise direct-matrix interpretation used here.

\section{Exact Support of Direct Matrix Evaluation}
\label{sec:direct}

We now state the paper's exact interface between a logical packed map and its
direct homomorphic schedule.  The theorem permits a different word matrix on
each $H$-fibre (the slots $h+s(Q)$ of one packed counter $h$), as occurs
when PASTA public material varies across packed
counters.

Let $\F$ be the plaintext field and $V\ne0$ an $\F$-vector space.
PASTA uses $V=\F$; the vector-space statement generalizes this scalar-word
interface. Identify packed coordinates with
\[
  G=\{h+s(y):h\in H,\ y\in Q\}.
\]
For each $h\in H$, let
$A^{(h)}=(A^{(h)}_{y,x})_{y,x\in Q}\in\F^{Q\times Q}$ be a word matrix.
Here $h$ identifies a packed counter in the contiguous instantiation,
permitting counter-specific matrices.
The fibrewise linear map $L_A:V^G\to V^G$ is
\[
  (L_Av)_{h+s(y)}=\sum_{x\in Q}A^{(h)}_{y,x}v_{h+s(x)}.
\]
Define its \emph{section difference support}
\[
  S_A(s)=\{s(y)-s(x):
        A^{(h)}_{y,x}\ne0\text{ for some }h\in H\}.
\]
Each nonzero matrix entry needs a movement from its source word to its
destination word. Entries with the same coordinate difference share one
coefficient mask; the theorem explains why a direct sum must retain every
required difference label.

\begin{theorem}
\label{thm:direct-support}
With $\tau_g$ oriented by $(\tau_gv)_z=v_{z-g}$, define, for
$g\in S_A(s)$, the scalar mask $d_g:G\to\F$ by
\[
 d_g(h+s(y))=
 \begin{cases}
 A^{(h)}_{y,x},&\text{if }g=s(y)-s(x)\text{ for some }x\in Q,\\
 0,&\text{otherwise,}
 \end{cases}
\]
and let $D_g:V^G\to V^G$ be the diagonal map
$(D_gv)_z=d_g(z)v_z$.  The section property makes the displayed $x$ unique.
Then $L_A$ has the exact aggregated normal form
\[
  L_A=\sum_{g\in S_A(s)}D_g\tau_g.
\]
After summands with the same $g$ have been aggregated, every
$g\in S_A(s)$ is necessary among representations of the form
$\sum_gD_g\tau_g$.
\end{theorem}

\begin{proof}
At destination $h+s(y)$, the $g$-summand reads
\[
 (\tau_gv)_{h+s(y)}=v_{h+s(y)-g}.
\]
For $g=s(y)-s(x)$ this is exactly $v_{h+s(x)}$, and the stated mask supplies
coefficient $A^{(h)}_{y,x}$.  Summing over labels therefore gives $L_A$.

For necessity, fix a nonzero entry $A^{(h)}_{y,x}$, put
$g=s(y)-s(x)$, and choose $0\ne v_0\in V$. Apply $L_A$ and any
representation $\sum_{g'}\widetilde D_{g'}\tau_{g'}$ that omits the label $g$
to the input supported only at $h+s(x)$, with value $v_0$ there. At destination
$h+s(y)$, every label $g'\ne g$ reads a zero source coordinate, so that
representation returns $0$, whereas $L_A$ returns $A^{(h)}_{y,x}v_0\ne0$.
Hence label $g$ cannot be omitted.
\end{proof}

\begin{corollary}
\label{cor:dense}
If, for every $(y,x)\in Q^2$, at least one fibre has
$A^{(h)}_{y,x}\ne0$, then $S_A(s)=T-T$.  Consequently the minimum number of
direct aggregated summands over all sections is $\delta(G,H)$, and the
minimum number of nonidentity summands is $\delta(G,H)-1$.
\end{corollary}

The density hypothesis concerns the union of nonzero entries across fibres,
not the matrix dimensions. A sparse workload can use a strict subset of $T-T$.
For a concrete PASTA corpus, full support depends on its generated coefficients;
the schedules below retain the entire stated support.

Proposition~\ref{prop:cyclic-square-transpose} applies the known formula for
cyclic quotients~\cite{BarcauPasolTurcas2026TDN}. Our contribution
here is its direct-matrix consequence: section search cannot beat the
obstruction, whereas exchanging the logical axes places word differences on a
kernel subgroup.

\begin{proposition}
\label{prop:cyclic-square-transpose}
Let $n\ge2$ be a power of two, let $R=C_{n^2}$, and let
$K=nR=\{0,n,2n,\ldots,(n-1)n\}\le R$, with arithmetic modulo $n^2$
(for $n=128$, $(R,K)$ is the pair $(G,H)$ of Section~\ref{sec:quotient-geometry}).
Then $K\simeq C_n$, $R/K\simeq C_n$, and
\[
  \delta(R,K)=2n-1.
\]
The interval $T_0=\{0,\ldots,n-1\}$ is an optimal section.  Put
$I_n=\{0,\ldots,n-1\}$.  On the logical representative grid
$I_n\times I_n$, the section-axis map into $R=C_{n^2}$
\[
 \iota_{\rm sec}(b,w)=nb+w
\]
therefore gives dense word support $T_0-T_0$ of size $2n-1$.  The
perfect-shuffle placement
\[
 \iota_{\rm ker}(b,w)=b+nw
\]
puts the word axis on $K$: at a fixed counter $b$, increasing the word
coordinate by one moves the slot by $n$. It gives dense word support $K$
of size $n$.
\end{proposition}

\begin{proof}
The prior cyclic formula gives $\delta(C_{\ell q},H)=2q-c$, where $H$ is
the subgroup of order $\ell$ (index $q$) and $c$ is the largest divisor of
$q$ coprime to $\ell$
~\cite[Corollary~3.6]{BarcauPasolTurcas2026TDN}. Here $\ell=q=n$, so $c=1$.
The interval has exactly the $2n-1$ signed differences
$-(n-1),\ldots,n-1$, attaining the bound. For the two placements, take
differences in the word coordinate: $w-w'$ for $\iota_{\rm sec}$ and
$n(w-w')$ for $\iota_{\rm ker}$. The latter range over the order-$n$
subgroup $K$. The single-source argument in the proof of
Theorem~\ref{thm:direct-support} also applies to $\iota_{\rm ker}$: for a
nonzero entry from word $x$ to word $y$ of counter $b$, only
$\tau_{n(y-x)}$ can read source $b+nx$ at destination $b+ny$. Under the same
density hypothesis across counters, every label in $K$ is therefore necessary
in direct form.
\end{proof}

There is no contradiction with the definition of $\delta(R,K)$.  The first
placement assigns the word dimension to quotient representatives and the
batch dimension to $K$; the transpose exchanges these roles.  It is not a
comparison of two sections for the same logical word axis.  Rather,
$\delta(R,K)$ certifies that changing the section alone cannot beat 255, which
motivates the kernel--section axis transpose.

The direct-form restriction is essential: a circuit may compose masked
translations, share intermediate results, and factor their sum. The next
section does exactly that. The theorem fixes which labels occur when such a
circuit is expanded into one aggregated sum; it does not lower-bound its
executed rotations, products, noise, key material, or time.

\section{Factoring Movements with BSGS}
\label{sec:bsgs}

Let $S$ be a translation support and suppose $S\subseteq B+J$.  Choose one
decomposition $g=b_g+j_g$ for every $g\in S$.  If
$D'_g=\tau_{-j_g}D_g\tau_{j_g}$, then
\[
  D_g\tau_g=\tau_{j_g}D'_g\tau_{b_g}.
\]
This identity justifies evaluating the baby rotations $\tau_bv$ once,
combining them with shifted masks, and applying giant rotations $\tau_j$ to
the partial sums.

For a cover $S\subseteq B+J$ containing zero in both $B$ and $J$, the
program uses $|B|-1$ nonidentity baby rotations and $|J|-1$ nonidentity
giant rotations. The same label may appear in both roles, so this sum counts
requests, not distinct evaluation keys. The elementary inequality
$|B+J|\le |B||J|$ gives a useful lower bound when all roles have unit cost.
For the full cyclic support $C_n$, it implies
\[
 |B|+|J|-2\ge \lceil2\sqrt n\rceil-2.
\]
At $n=128$, the bound is 21. The sets
\[
 B=\{0,\ldots,9\},\qquad J=\{0,10,\ldots,120\}
\]
attain it. Assign $d\in\{0,\ldots,127\}$ to
$b_d=d\bmod10$ and $j_d=10\lfloor d/10\rfloor$. This uses 128 of the
130 cover cells and gives $9+12=21$ nonidentity roles. This is an exact optimum
for the uniform role-count objective only; it need not minimize a cost that
weights baby operations, giant operations or physical key switches differently.

For Proposition~\ref{prop:cyclic-square-transpose} at $n=128$, the
contiguous-block word support is the signed interval $[-127,127]$ in
$C_{16384}$.  It has the explicit cover
\[
 B=\{-7,\ldots,7\},\qquad J=\{-120,-105,\ldots,120\}.
\]
For $d\in[-127,127]$, set
$q=\lfloor(d+7)/15\rfloor$, $j_d=15q$, and $b_d=d-j_d$; then
$q\in[-8,8]$, $b_d\in[-7,7]$, and $d=b_d+j_d$.  Hence this complete signed
$15\times17$ cover has cost $14+16=30$, which matches the product lower bound
for 255 labels.  After the perfect shuffle the word
support is $C_{128}$ and the displayed cover uses 21 roles.  The difference is exact for the
uniform role-count objective.  It does not by itself predict the time ratio:
the layouts use different masks and automorphism values, and the backend may
assign different paths and costs to their roles.

\section{Three Equally Batched PASTA-3 Implementations}
\label{sec:implementation}
\label{sec:primary-transpose}

The BGV context has $m=65536$, $p=65537$, and $r=1$, giving two cyclic
rows of length $L=16384$. For branch $\epsilon\in\{0,1\}$, row
position $u$ has physical slot $2u+\epsilon$. Row rotations act on both
branches simultaneously; Mix uses a separate row swap. Section~\ref{sec:setup}
specifies the remaining parameters and library configuration.

Both contiguous candidates place counter $b$, word $w$ at $128b+w$.
The subgroup candidate places it at $b+128w$. Each receives a fresh HE
encryption of the same key repeated in its counters' state slots. A
permutation between these ciphertexts is not free, so the comparison starts
from native inputs and specifies native outputs.

\subsection{The affine schedules}

The direct contiguous schedule uses signed word differences
$d\in\{-127,\ldots,127\}$. Writing $A_{\epsilon,b}$ for the current
layer's public matrix on branch $\epsilon$ of counter $b$, its mask at
destination $128b+w$ is
\[
 E_d[\epsilon,128b+w]=
 \begin{cases}
 A_{\epsilon,b}[w,w-d],&0\le w-d<128,\\
 0,&\text{otherwise}.
 \end{cases}
\]
The zeroes suppress crossings into adjacent counters. A signed $15\times17$
BSGS cover uses babies $-7,\ldots,7$ and giants $-120,-105,\ldots,120$.
Each diagonal appears exactly once, with its coefficient mask transported by
the negative giant rotation. This executes 30 nonidentity row rotations and
255 coefficient products for the two branches together.

For contiguous composed evaluation, let $P_s$ rotate each 128-word block
locally. For $0<s<128$, the four-word example generalizes to
\[
 P_s x=\tau_s x+M_s\odot(\tau_{s-128}x-\tau_s x),\qquad
 M_s[\epsilon,128b+w]={\bf1}_{w<s}.
\]
The mask selects the source in the same block, including at the first and last
blocks of the row. Reduce any signed shift modulo 128; a zero shift is the
identity. The resulting permutations satisfy $P_sP_t=P_{s+t\bmod128}$, so
they support the same BSGS factorization as ordinary cyclic rotations.

Use babies $t=0,\ldots,9$ and giants $j=0,10,\ldots,120$, retaining only
$j+t<128$. At position $128b+w$ the prepared coefficient for cell $(j,t)$ is
\[
 \widehat D_{j,t}[\epsilon,128b+w]
 =A_{\epsilon,b}[(w+j)\bmod128,(w-t)\bmod128].
\]
Accumulate $\widehat D_{j,t}\odot P_t x$ inside each giant group, then apply
$P_j$. There are 21 nonidentity local rotations, each requiring two row
rotations and one wrap-mask product in this program.

For subgroup evaluation, replace $P_s$ by $Q_s=\tau_{128s}$. At row coordinate
$b+128w$, this reads word $w-s\bmod128$ of counter $b$. The prepared
coefficient formula is identical after changing the placement. No wrap repair
is needed. The three schedules therefore have the counts in
Table~\ref{tab:arms}. Constants are added once after each matrix result, and
all candidates copy the first term into an accumulator instead of adding it to
an artificial zero.

\begin{table}[t]
\centering
\caption{One affine layer on both branches and 128 counters. Width is the
number of baby steps $|B|$. Counts exclude
branch Mix and refer to program operations, not physical key switches.}
\label{tab:arms}
\small
\begin{tabularx}{\linewidth}{@{}Yrrr@{}}
\toprule
Schedule & Row rotations & Coefficient products & Wrap products\\
\midrule
\tablerows{tables/table1-affine-count-rows.tex}
\bottomrule
\end{tabularx}
\end{table}

The composed schedule illustrates the scope of the direct-support theorem.
When expanded into a single direct sum, it still has the required direct
support. Its factored program shares intermediate work and has fewer plaintext
products. Its wrap-mask multiplications and additional rotations can also
change preparation time and noise consumption. These costs must be inspected
and measured, rather than inferred from the 255-versus-128 support sizes.

\subsection{Nonlinear layers and returned slots}
\label{sec:nonlinear}

Each affine result enters Mix, implemented by a branch swap and two
ciphertext additions. For a predecessor-square step, let $M$ be one at
words $1,\ldots,127$ in both branches and zero at word zero. The row
displacement is $s=1$ for either contiguous schedule and $s=128$ for subgroup
placement. We execute
\begin{equation}
 x+M\odot\tau_s(x^2).
 \label{eq:nonlinear-order}
\end{equation}
Squaring a copy of $x$, rotating it, applying $M$, and adding it to $x$
uses one square, one row rotation, one plaintext product, and one addition.
The mask suppresses the unwanted predecessor at word zero. For all other
words, the rotation reads the correct predecessor without a repaired local
cyclic rotation. Since $M^2=M$ and translation preserves slotwise products,
this equals $x+(M\odot\tau_s x)^2$. Thus square--rotate--mask and
rotate--mask--square implement the same PASTA transformation. All complete
conversions reported below use square--rotate--mask. Their equivalence does
not imply equal noise: modulus selection, rounding, and relinearization
intervene between the operations.

The componentwise cube uses two ciphertext products, each followed by
relinearization. After the final Mix, one plaintext mask keeps the left
branch and clears the right; negation and addition of the encoded public
ciphertext then return the messages. This enforces the output
specification in Section~\ref{sec:background}. Sanitizing unused slots is part of
evaluation because they would otherwise expose additional state to a
recipient who decrypts the returned ciphertext.

\subsection{Consuming a converted batch with a later query}
\label{sec:consumer}

The server retains the sanitized ciphertext and accepts a public query
$\alpha\in\{0,1\}^{128}$ after conversion. A query uses a copy of that
ciphertext, so two different queries need neither another PASTA evaluation
nor query-specific PASTA matrices. Its intended result is
\[
 y_b=\sum_{w=0}^{127}\alpha_wv_{b,w}.
\]
For the measurement example only, $0\le v_{b,w}\le255$; hence
$0\le y_b\le32640<65537$. General cipher tests still use field-valued
messages. The application restriction gives the sum an integer interpretation;
it is not a change to PASTA or its plaintext field.

Let $d=1$ for contiguous placement and $d=128$ for subgroup placement.
Multiply the input by a public mask containing $\alpha_w$ at message positions
and zero in the other branch, obtaining $S_0$. Compute
\[
 S_{j+1}=S_j+\tau_{-d2^j}S_j\quad(0\le j<7).
\]
Induction gives $S_j=\sum_{t=0}^{2^j-1}\tau_{-dt}S_0$.
At row position $128b$ in the contiguous layout, $S_7$ therefore reads exactly
positions $128b,\ldots,128b+127$. None crosses a record boundary, including
for the last record. At position $b$ in subgroup placement it reads exactly
$b+128w$ for $0\le w<128$. Other contiguous destinations can contain sums
crossing record boundaries; we do not return them. A final mask keeps only
branch zero at the designated positions, enforcing
\begin{equation}
 \mathrm{queryout}_{2\iota(b,0)}=y_b,
 \qquad \mathrm{queryout}_{k}=0
 \quad\text{for every other slot }k.
 \label{eq:query-output}
\end{equation}
Thus both layouts need seven ordinary row rotations, seven additions, and two
plaintext products. Repairing a cyclic block rotation at every tree level
would do unnecessary work. The two contiguous cipher schedules share this
query algorithm; neither layout requires a transpose.

The query algorithm receives the retained ciphertext and public weights,
without clear messages or the secret key. Each query copies the retained
conversion: its noise growth affects its own answer, not the stored batch or
another independent query. Two tested queries therefore do not impose a
maximum of two queries on the batch. A same-context control applies the
algorithm to freshly encrypted records in each native layout, isolating query
arithmetic from noise accumulated during conversion. Plaintext differential
and boundary tests check every slot of the reduction and selected output.
Section~\ref{sec:query-results} compares fresh and converted inputs at two
modulus choices, including both successful queries and capacity failures.

\section{Evaluation: Conversion and Later Queries}
\label{sec:newexperiment}

\subsection{Experimental setup and measurement boundaries}
\label{sec:setup}

The paired study uses
$(m,p,r,\mathrm{bits},c)=(65536,65537,1,460,2)$, with
128 records of 128 words per batch and no bootstrapping. We use
HElib~\cite{HElibSourceV230}, NTL and GMP at the versions recorded in the
repository~\cite{PASTAPackingCode2026}. The executable is built for arm64
with AppleClang 21.0.0, C++17.
The host is an Apple M5 with ten CPU cores and 24 GiB RAM, running
macOS 26.6.2. NTL uses one thread. Candidates run sequentially,
without concurrent benchmark, build or document-generation jobs; ordinary
desktop and operating-system activity is uncontrolled.

Calibration compares direct widths 15, 16 and 17 with subgroup widths 10,
11 and 16 on two corpora under separate fresh HE keys. All six covers pass
complete conversion and both queries on both corpora. The fixed selection
rule minimizes median recurring server cost within each layout, including
fresh public generation, conversion and two queries; ties favor the smaller
width. It selects direct width 17 and subgroup width 11. Direct width 17
uses babies $-8,\ldots,8$ and giants $17k$, $-7\le k\le7$;
subgroup width 11 uses babies $0,\ldots,10$ and giants $11k$, $0\le k\le11$,
omitting cells above 127. Their counts equal the corresponding reference
counts in Table~\ref{tab:arms}. Direct width 16 uses 31 rotations.
The winning direct median is only 16 ms (0.043\%) below width 15;
selection does not establish a meaningful advantage for that particular width.

The selected programs receive twelve paired corpora across six new HE keys,
two corpora per key, with alternating layout order. Calibration and measurement
fixtures come from disjoint deterministic domains. Each pair shares its
logical PASTA key, byte-valued records, nonce/counter material and queries;
these vary across pairs. Counters are $0,\ldots,127$. Records zero and 127
contain all zeroes and all 255s. Each batch receives two independently evaluated
queries on copies of its retained conversion. The declared query profiles
combine random Boolean subsets with the full subset and endpoint singletons.
The repository includes the fixture generator, which uses seed 20260910;
its reproducible keys are test inputs, not a production key-generation method.

One context is reused within each phase. Both layouts use the same HE key
and the union of their required automorphism keys plus relinearization material;
each gets a separate encryption of the repeated PASTA key in its native layout.
Each selected schedule owns its reusable geometry encodings: the plaintext
encodings of the predecessor mask, output-branch mask and query selector,
which depend on the layout and context rather than nonce/counter material.
New nonce/counter material always incurs public expansion and coefficient preparation, with only
one dense affine layer's encodings live at once. Public generation is measured
once per corpus and charged equally to each alternative.

Outer monotonic timers cover complete conversion and complete query calls,
including coefficient/weight preparation, required ciphertext copies,
arithmetic and output zeroing. The recurring server interval is public
generation plus conversion plus two queries. Inner preparation and arithmetic
timers explain this cost. Tables~\ref{tab:conversion}
and~\ref{tab:query-cost} also provide conversion-only and per-query costs.
Context/key/geometry setup, input encryption,
serialization, decryption, all-slot checking, capacity sampling and file writes
are outside this interval. HElib's plaintext representation expansion and warning
formatting remain inside the arithmetic intervals; warning-file writes are outside.
We summarize medians, inclusive linearly interpolated quartiles and ranges.
These are descriptive observations, with six independent HE keys; neither
queries nor slots are independent key trials. No confidence interval or
failure-probability estimate is asserted. No measured pair was replaced.

HElib's remaining noise capacity is the logarithmic margin between the active
modulus and its modeled total-noise bound, reported as integer bits:
$\log_2 Q_{\mathrm{act}}-\log_2\mathcal{N}$. It is neither a security level
nor a measured plaintext error. Modulus changes during multiplication and key
switching mean that a noise log alone does not determine capacity loss.
Separate single-fixture diagnostics at 400 and 460 bits show the composed
schedule's refusals and the effect of insufficient query headroom
(capacity left for queries after conversion).
Their fixed order and shared caches differ from this protocol; their timings
are not pooled with the paired study.

\subsection{Correct conversion and equally batched cost}
\label{sec:conversion-results}

Plaintext tests check the published PASTA known-answer vectors
\cite{PASTAArtifact,PASTAFramework}, independently indexed affine and
full-cipher schedules, physical boundaries, inactive slots and both predecessor
orders. Consumer tests include empty/full subsets, endpoint words and
byte-boundary records. The field codec also checks every value in $\F_{65537}$.
In the paired encrypted study, all 24 conversions recover all 16,384 message
words and all 16,384 required zeroes. All 48 later queries recover the expected sums and required zeroes, and serialization checks confirm that their retained inputs remain
unchanged. Calibration separately completes twelve conversions, 24 converted
queries and eight same-context fresh-BGV control queries.

\begin{table}[t]
\centering
\caption{Paired recurring server cost in seconds: median [first, third
quartile] over twelve corpora per layout. Complete conversion includes its
preparation and arithmetic; the final row includes public generation,
conversion and two queries. Medians of components need not add.}
\label{tab:conversion}
\small
\begin{tabularx}{\linewidth}{@{}Yrr@{}}
\toprule
Interval & Contiguous direct & Subgroup aligned\\
\midrule
\tablerows{tables/table2-conversion-rows.tex}
\bottomrule
\end{tabularx}
\end{table}

The median recurring server costs are 36.67 seconds direct and 22.91 seconds
subgroup (Table~\ref{tab:conversion}). The median of the twelve paired
direct/subgroup ratios is 1.602, with quartiles 1.597--1.603 and range
1.573--1.632. This is a measured advantage for the selected, equally batched
programs under the declared cost boundary. Preparation accounts for most of
the conversion cost. New counter-dependent matrices cannot be treated as free
constants, and the support count alone does not predict this ratio.

Setup is separate: measured-phase context construction takes 33.49 seconds,
and geometry encoding takes 2.73 seconds direct and 2.35 seconds subgroup.
Across six keys, median key generation and evaluation-key generation take
35 ms and 1.33 seconds. The common union contains 62 automorphism matrices
plus relinearization material and serializes to approximately 401.1 MB.
Peak resident set size (RSS) is 9.93 GiB for the measured phase.
We do not claim that the shared-resource observations
establish standalone key or memory savings for either layout.

\subsection{Capacity and the ability to answer later queries}
\label{sec:query-results}
\label{sec:noise}

Each answer has 128 permitted sums, at physical slots $256b$ for contiguous
placement or $2b$ for subgroup placement, and zeroes in every other slot.
Direct query capacities range from 33 to 39 bits; subgroup capacities range
from 22 to 32 (Table~\ref{tab:query-diagnostic}). The subgroup layout thus retains less
headroom despite its lower recurring cost. Some answers fall below the
30-bit engineering reserve considered during parameter selection; this reserve
is not a correctness threshold or security level. Success on these inputs
does not bound the failure probability for other keys or public material.

\begin{table}[t]
\centering
\caption{Capacity in bits and outcomes after checking every slot. Paired rows give ranges over
12 conversions and 24 queries per layout. Prior rows are separate
single-fixture observations; their two queries have the same displayed
capacity and outcome. A refusal yields no recovered plaintext.}
\label{tab:query-diagnostic}
\small
\begin{tabularx}{\linewidth}{@{}lYll@{}}
\toprule
Bits, study & Input path & Conversion & Later queries\\
\midrule
\tablerows{tables/table3-capacity-outcome-rows.tex}
\bottomrule
\end{tabularx}
\end{table}

The composed width-10 program executes its arithmetic graph but is refused
by HElib's noise guard in both single-fixture contexts
(400 and 460 bits, Table~\ref{tab:query-diagnostic}). This is an adverse result for
that program and those parameters, not an impossibility result for composed
schedules or contiguous packing. At 400 bits, direct and subgroup conversion
pass but their later queries are refused. Fresh same-context controls pass,
with capacities 320 and 314, respectively, locating the difficulty in the
margin left after conversion.

Public masks help explain these noise costs. A Boolean slot mask need not have
polynomial encoding norm one. For wrap mask $M_s$ from
Section~\ref{sec:implementation}, let $m_s(X)$ be its balanced polynomial
encoding. HElib's embedding routine computes
$\max_\zeta|m_s(\zeta)|$ over primitive 65,536th roots of unity.
The public sizes are approximately $2^{19.98}$ and $2^{19.75}$ for wrap
steps 1 and 10. Such masks recur before nonlinear stages amplify accumulated
noise. Query output selectors have sizes approximately $2^{19.32}$ contiguous
and $2^{23.59}$ subgroup. Avoiding wrap repair does not guarantee a cheaper
query selector; these observations do not allocate every noise cost.

\subsection{Query cost and returned ciphertexts}
\label{sec:query-cost}

The sum of two complete query calls has median cost 0.684 seconds direct
and 0.708 seconds subgroup. Thus the recurring-cost advantage is not evidence
of faster later queries. Table~\ref{tab:query-cost} gives the individual-query
intervals, whose variation includes different public masks. Decryption and
all-slot decoding take about 7.2 seconds per answer on this host and are
outside the server interval. Network transfer and service scheduling are not
measured; server arithmetic alone is not the cost of obtaining a decoded answer.

\begin{table}[t]
\centering
\caption{Individual query costs in milliseconds: median [first, third
quartile] over 24 queries per layout. Serialization and decryption/decoding
are separate from the complete server query. Oracle comparison is excluded.}
\label{tab:query-cost}
\small
\begin{tabularx}{\linewidth}{@{}Yrr@{}}
\toprule
Interval & Contiguous direct & Subgroup aligned\\
\midrule
\tablerows{tables/table4-query-rows.tex}
\bottomrule
\end{tabularx}
\end{table}

Ten direct conversions produce ciphertexts of 4,194,800 serialized bytes;
the other two occupy 3,670,472 bytes. Their respective query answers have
the same sizes.
All subgroup conversion ciphertexts and answers occupy 3,670,472 bytes.
Active modulus sets can change during evaluation, so serialized sizes need
not be constant across keys. These are actual serialized objects, without
compression. Selecting 128 answer slots still returns a full-ring ciphertext.

\subsection{Client work, reuse, and direct ingestion}
\label{sec:client}

The shallow direct-BGV comparator uses the smaller context
\[
 (m,p,r,\mathrm{bits},c)=(32768,65537,1,160,2),
\]
with all 16,384 slots holding useful words.
Physical slot $2(128(b\bmod64)+w)+\lfloor b/64\rfloor$ stores record $b$,
word $w$; each row holds 64 contiguous records. Weighting, seven rotate-and-add
levels and a selector target record $b$'s sum at
$256(b\bmod64)+\lfloor b/64\rfloor$, with every other slot zero.
This matches the logical records and permitted answers without an extra
cipher-state branch. Its mapping and reduction pass plaintext tests.

The encrypted baseline stops during its first query with a library exception
in ciphertext prime-set handling, before producing an answer. The failing internal operation has not been localized. This is an unresolved
implementation failure, distinct from the composed schedule's noise refusal.
It leaves baseline query cost, answer size, correctness and a complete
direct-ingestion comparison unresolved. The attempt does record one input:
encoding/encryption takes 8.00 ms, serialization 0.65 ms, and its upload
occupies $U=786,728$ bytes. These are partial observations from one fixture.

Across the twelve measured corpora, native PASTA encryption takes a median
103.97 ms, including nonce/counter expansion and field addition; the outer
producer interval including wire encoding is 104.02 ms. These timings do not
support a producer-compute advantage over the single shallow-BGV observation.
The implemented upload stores each field element in 17 bits, including 65536,
with a 16-byte field header and a 32-byte envelope for the shared nonce,
counters and dimensions. Without compression, each batch occupies
$R=16+32+17(128\cdot128)/8=34,864$ bytes. All 65,537 field values round-trip;
truncated and noncanonical encodings are rejected.

The HE-encrypted PASTA key occupies $K=4,194,800$ bytes in both layouts;
median encoding plus encryption takes 43.16 ms direct and 44.84 ms subgroup.
For $E$ batches per symmetric-key epoch, record uplink averages $R+K/E$,
versus $U$ for direct ingestion. It is smaller when $U>R$ and
$E>K/(U-R)$. The observed sizes give a threshold of six batches and costs
4,229,664, 454,344 and 76,812 bytes per batch for $E=1,10,100$.
These are illustrative byte-accounting scenarios using the partial baseline
input, not evidence that its query path works. Distinct nonce/counter pairs
are required under each PASTA key~\cite[Section~7]{PASTA}. The scenarios use
at most 12,800 block identifiers; identifier availability is not a new
cryptographic usage bound. The original cipher assumption remains necessary.

HE public/evaluation keys have a separate reuse interval. If their distribution
cost is $D$ over $E'$ batches under one HE key, account for $D/E'$ in its actual
owner-to-producer or owner-to-server direction, separately for each method;
$E'$ need not equal $E$. The paired key union occupies about 401.1 MB; the
partial shallow baseline's public/evaluation bundle occupies 10.23 MB.
These differ in context and functionality. Producer compute would favor PASTA
only if $C_{\rm PASTA}+C_{\rm provision}/E<C_{\rm direct\ HE}$; the available
timings point the other way. The server additionally pays conversion, and
answers remain full-ring ciphertexts. Reduced uplink does not establish
lower total latency. Repeated queries on one retained batch do not amortize
key provisioning over additional record batches.

\subsection{Parameters and security}
\label{sec:security}

We use the lattice estimator~\cite{LatticeEstimator2026} with an ordinary
learning-with-errors (LWE) proxy: dimension $N$, the actual modulus $QP$
including special key-switch primes, unlimited independent samples,
unconditioned ternary secret probabilities $(1/4,1/2,1/4)$ on $(-1,0,1)$,
and discrete-Gaussian error scale 3.2. HElib forms a zero-message key relation
$b=-as+pe$ modulo the ciphertext modulus, where $a$ is uniform, $s$ is the
secret and $e$ is sampled at scale 3.2~\cite{HElibSourceV230}.
Since $p=65537$ is invertible modulo that modulus, multiplication by $p^{-1}$
gives $p^{-1}b=-(p^{-1}a)s+e$: the scaled $a$ stays uniform and the remaining
error is $e$. This motivates the proxy's error scale. The proxy does not exactly
model ring correlations, HElib's polynomial conditioning or secret-dependent
evaluation-key messages. The original PASTA and evaluation-key assumptions
remain necessary~\cite{PASTA,BossuatEtAl2025HEGuidelines}.

For conversion, $N=32768$, $\log_2 Q=463.338269$,
$\log_2 P=207.449098$ and $\log_2(QP)=670.787367$.
Completed classical searches include primal unique shortest vector problem
(uSVP), bounded distance decoding (BDD), dual, dual-hybrid and primal BDD-hybrid
attacks. Here MATZOV refers to the cited report~\cite{MATZOV2022}, and
ADPS16 to Alkim et al.'s classical
core shortest-vector-problem (core-SVP) reduction-cost model
\cite{AlkimEtAl2016NewHope}, with cost $2^{0.292\beta}$ at block size $\beta$;
the latter names a cost model, not an attack family. Minimum estimated log-work
is 168.60 under MATZOV, attained by BDD, and 137.53 under ADPS16, attained by
dual-hybrid.

For shallow direct ingestion, $N=16384$ and $\log_2(QP)=243.457043$.
The completed uSVP/BDD/dual/dual-hybrid minima are 240.28 (MATZOV, BDD) and
212.87 (ADPS16, dual-hybrid). An additional ADPS16 primal BDD-hybrid estimate
is 214.05. These results do not validate the failed query implementation.
Software versions, commands, modulus components and sensitivity results at error
scale 2.8 are in the repository~\cite{PASTAPackingCode2026}. This perturbation is
not a proved conservative bound.

The full attack suite was not completed because of computational cost:
additional meet-in-the-middle hybrid and coded-BKW estimates remain incomplete,
and Arora--Gr\"obner was not evaluated. These are conditional estimates for
the stated proxy, not lower bounds on attack cost or a certificate of 128-bit
protocol security.

\section{Related Work}
\label{sec:related}

Homomorphic SIMD packing and HElib's diagonal, BSGS and hoisting methods
provide the basic arithmetic and movement tools used here
~\cite{SmartVercauteren2014SIMD,HaleviShoupAlgorithmsInHElib,HaleviShoup2018LinearTransforms}.
PASTA uses diagonal/BSGS layers and packs its two branches in parallel rows.
Its word-sliced alternative avoids rotations but changes ciphertext organization
and the work needed for packed output
\cite[Sections~6.5.2--6.5.3 and Appendix~C]{PASTA}.
Our equally batched layouts retain the packed cipher and use the authors' public
material and reference implementation~\cite{PASTAArtifact,PASTAFramework}.

Layout-aware HE compilation likewise connects placement to subsequent computations.
CHET records physical strides, logical dimensions and padding in ciphertext
tensors, chooses layouts by estimated cost, and delays costly reshaping through
metadata updates~\cite[Sections~4.2 and 5.3]{CHET}. Porcupine synthesizes
vectorized kernels from instruction semantics and a latency/noise objective;
its dot-product example uses rotate-and-add reduction
~\cite[Sections~3.2, 4.2 and 5.2]{Porcupine}. Both motivate comparing executable
schedules and subsequent computations instead of treating support count as a latency model.
We specialize established packing and scheduling tools to complete PASTA
conversion and subsequent queries.

Peikert and Pepin's top-down Chinese remainder theorem (CRT) construction uses a transversal to constrain
possible automorphism coefficients by its difference set
~\cite{PeikertPepinViveGalois1}. Barcau, Pa\c{s}ol, and \c{T}urca\c{s}
define the invariant and prove the cyclic formula
\cite{BarcauPasolTurcas2026TDN}. Theorem~\ref{thm:direct-support} connects
counter-dependent matrices to the direct form $\sum_gD_g\tau_g$.
The complete-cipher schedules
and subsequent queries then expose costs and noise
failures that direct support alone does not predict.

\section{Discussion and Conclusion}
\label{sec:limitations}
\label{sec:conclusion}

Theorem~\ref{thm:direct-support} connects counter-dependent matrix evaluation
to direct translation support. At full batch size, dense word maps require
255 direct labels in contiguous placement and 128 in subgroup placement.
Composed rotations share
work despite the larger expanded support, but add masks whose costs propagate
through the nonlinear cipher. The choice of packing therefore depends
on mask preparation and subsequent computation as well as on rotation counts.

In the measured study, all 24 conversions and 48 subsequent queries recover
the specified values and required zeroes. Across twelve paired corpora under
six HE keys, the median direct/subgroup recurring server-cost ratio is 1.60,
including fresh public generation, conversion and two queries. The subgroup layout
reduces this conversion workload's cost while leaving less query headroom;
the evidence does not establish faster subsequent queries.

The tested composed schedule is refused by the noise guard in both single-fixture
contexts (400 and 460 bits), and the smaller 400-bit context supports neither layout's converted
queries. These observations do not rule out other composed schedules or bound
a decryption-failure probability. The shallow direct-ingestion baseline's
first-query failure leaves the end-to-end application comparison open, and
partial producer measurements support no PASTA compute advantage. The security
estimates remain conditional on the stated proxy and incomplete attack coverage.
Within these limits, these results establish a packing-dependent cost--noise
tradeoff.

\section{Code and Data Availability}
\label{sec:artifact}

Code, dependencies and recorded results are
available in the GitHub repository ~\cite{PASTAPackingCode2026} together with instructions for installation, testing and reproduction of the results.

\section*{Use of Artificial Intelligence Tools}

OpenAI's Codex (multiple models) was used to assist with the implementations,  manuscript drafting and revision.

\bibliographystyle{alpha}
\bibliography{references}
\end{document}